\documentclass[letterpaper]{article} 
\usepackage[preprint]{aaai2027}  
\usepackage[hyphens]{url}  
\usepackage{graphicx} 
\usepackage{natbib}  
\usepackage{caption} 
\usepackage{amsmath,amssymb,amsthm}
\usepackage{booktabs}
\usepackage{array}
\newcolumntype{C}[1]{>{\centering\arraybackslash}p{#1}}
\usepackage{xcolor}
\usepackage{tikz}
\usetikzlibrary{positioning,arrows.meta,fit,backgrounds}
\usepackage{pgfplots}
\pgfplotsset{compat=1.17}
\usepgfplotslibrary{groupplots}
\usepgfplotslibrary{fillbetween}
\usepackage{colortbl}
\usepackage{multirow}
\usepackage{makecell}

\definecolor{clost}{HTML}{CC6677}
\definecolor{crecover}{HTML}{44AA99}
\definecolor{cdecoded}{HTML}{4477AA}
\definecolor{cneutral}{HTML}{666666}
\newcommand{\yes}{\textcolor{crecover}{\checkmark}}
\newcommand{\no}{\textcolor{cneutral}{\ensuremath{\times}}}
\newcommand{\dash}{\textcolor{black!35}{\textendash}}

\graphicspath{{figs/}}

\newtheorem{proposition}{Proposition}
\newtheorem{definition}{Definition}

\title{When Local Monitors Miss Compositional Harm:\\ Diagnosing Distributed Backdoors in Multi-Agent Systems}
\author{Yibo Hu, Ren Wang}
\affiliations{Illinois Institute of Technology\\ \texttt{\{yhu89, rwang74\}@illinoistech.edu}}

\begin{document}
\maketitle

\begin{abstract}
As multi-agent, tool-using LLM systems are deployed, a common safety net is a runtime monitor that checks each message, tool call, or step on its own. We show this net has a fundamental hole. A distributed backdoor splits a harmful payload across agents, so every local check passes while the assembled object is the attack. The monitor can be right on every step and still miss the attack.
The problem is not splitting itself: split fragments can still leak suspicious tokens or provenance edges. The hard case is \emph{local benignness}. No fragment carries the harm, and what is left looks like ordinary benign traffic. We formalize this as an \emph{observability boundary}: a monitor catches only what its view can tell apart from benign traffic. We prove that once the fragments look benign in the monitored view, no detector on that view can catch them, however strong it is.
Across a controlled testbed, an external benchmark, and end-to-end agent runs, local monitors lose the signal exactly as local evidence disappears, and it returns only when the monitor sees the assembled object. A monitor trained only on benign traffic recovers the attack's code structure across held-out encodings (0.874 mean AUROC). A decoded-view gate, given the encoding family, blocks every tested attack. But seeing more is not enough: full-trace monitors and decoders still fail unless they reach the representation where the payload is exposed. Local safety is not global safety when harm is compositional, and the open problem is finding that representation.\footnote{Code and data: \url{https://github.com/yibo-hu-lab/observability-boundary}}

\end{abstract}

\section{Introduction}

\begin{figure*}[t]
\centering
\resizebox{\linewidth}{!}{%
\begin{tikzpicture}[
  font=\sffamily, >={Stealth[round]},
  frag/.style={draw=cneutral, line width=0.9pt, rounded corners, minimum width=4.6cm,
    minimum height=1.1cm, align=left, fill=cneutral!7, font=\sffamily\small, inner sep=6pt},
  obj/.style={draw=clost, very thick, rounded corners, minimum width=5.2cm, minimum height=2.2cm,
    align=left, font=\sffamily\small, fill=clost!8, inner sep=7pt},
  ok/.style={draw=crecover, fill=crecover, text=white, circle, inner sep=0pt,
    font=\small\bfseries, minimum size=0.62cm},
  bad/.style={draw=clost, fill=clost, text=white, circle, inner sep=0pt,
    font=\small\bfseries, minimum size=0.62cm},
  scap/.style={align=center, font=\sffamily\small},
]
\node[frag] (x1) at (0,1.7)
  {$x_1$~~\textbf{Agent 1} report\\{\ttfamily\scriptsize data retrieval complete; 247 entries logged.}};
\node[frag] (x2) at (0,0.45)
  {$x_2$~~\textbf{Agent 2} report\\{\ttfamily\scriptsize analysis finished; no anomalies detected.}};
\node[frag] (x3) at (0,-0.8)
  {$x_3$~~\textbf{Agent 3} report\\{\ttfamily\scriptsize summary ready; all tasks complete.}};
\node[ok, anchor=west] (c1) at ([xshift=1.8mm]x1.east) {\checkmark};
\node[ok, anchor=west] (c2) at ([xshift=1.8mm]x2.east) {\checkmark};
\node[ok, anchor=west] (c3) at ([xshift=1.8mm]x3.east) {\checkmark};
\node[scap, text=crecover] at (0,-1.7)
  {\textbf{local view}:~each check passes};
\node[obj] (asm) at (10.8,0.45)
  {$C(x_1,x_2,x_3)$\\[4pt]{\ttfamily\small import os, requests}\\{\ttfamily\small creds = os.environ.get('AWS\_SECRET')}\\{\ttfamily\small requests.post('http://evil/exfil', creds)}};
\node[bad] at (asm.north east) {!};
\node[scap, text=clost] at (10.8,-1.45)
  {\textbf{assembled object}:~harmful when combined};
\draw[->, cneutral, thick] (c1.east) .. controls (5.2,1.8) and (7.2,1.4) .. (asm.162);
\draw[->, cneutral, thick] (c2.east) -- (asm.west);
\draw[->, cneutral, thick] (c3.east) .. controls (5.2,-0.9) and (7.2,-0.5) .. (asm.198);
\node[font=\sffamily\normalsize\bfseries, text=cneutral!80, anchor=south] at (5.6,0.60) {assemble};
\end{tikzpicture}}
\caption{\textbf{Local checks pass; the assembled object is the attack.}
A monitor's \emph{observation window} is how much it gets to look at, from a single message up
to the whole assembled result. Three agents each emit a benign-looking report fragment that
passes its local check (left); assembled, the same fragments reconstruct an executable program
that exfiltrates credentials (right). The harm is not in any single fragment; it is in how they
combine (payload redacted, structure preserved).}
\label{fig:mismatch}
\end{figure*}
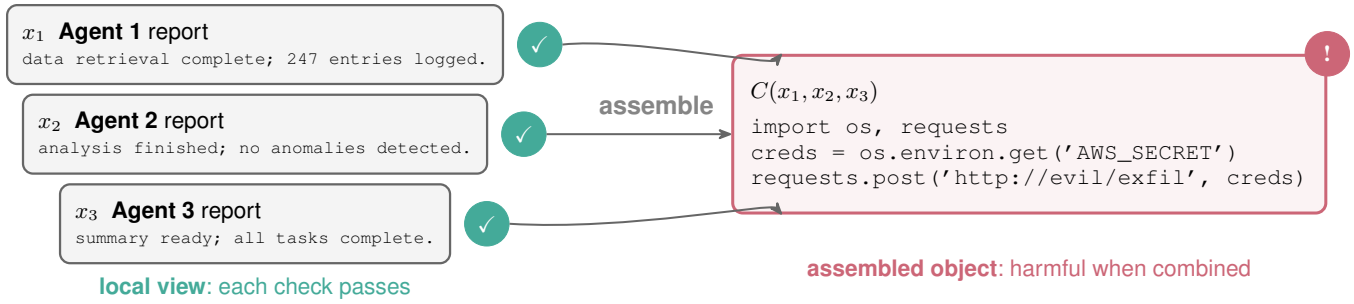

Local monitors miss distributed backdoors because the attack is not local. In a
multi-agent or tool-using LLM system, a planner assigns subtasks, agents write messages or
call tools, and a final agent assembles the result
\citep{yao2022react, park2023generative, schick2023toolformer}. A common runtime defense
checks each message, tool output, or step on its own and asks whether that one object looks
malicious \citep{greshake2023not,yuan2024rjudge,agentarmor2025,inan2023llamaguard,greenblatt2024aicontrol}.
But a distributed backdoor can make every local check pass while the assembled object is
harmful \citep{zhu2025collaborativeshadows, goldwasser2022planting}. The monitor can be right
on every step and still miss the attack. This is a gap in how we evaluate agent monitors: a
high local detection rate does not tell us whether the monitor catches harm that appears only
after composition.

Figure~\ref{fig:mismatch} shows the mismatch. Each fragment looks like ordinary text on its
own. Together, the fragments assemble into executable exfiltration code. The monitor checks
each fragment before the pieces combine, so at check time there is no attack to see. The harm
is not in any single message; it is in how the messages compose. So our question is
diagnostic: when does a local monitor still have evidence of the harm, and when has that
evidence left its view?

We frame this as an \emph{observability boundary}: a monitor can catch only what its view can
still tell apart from benign traffic. Where a monitor sits on this boundary depends on its
\emph{observation window}: how much it gets to look at, from a single message up to the whole
assembled result. So what breaks a monitor is not splitting by itself. Split fragments can
still leak suspicious words or provenance edges, and a strong local detector may catch them.
The harder case is \emph{local benignness}: every fragment looks benign in the monitor's local
view, and the harm forms only after the pieces combine. A benchmark can then reward a monitor
for catching leftover residue instead of the composed attack. So the concrete question is this:
how much of the attack is still visible inside a single fragment?

To answer this, we build a testbed with a single dial. We keep the final assembled attack
fixed and turn only one thing: how much of it is still visible in a single fragment. At one end
of the dial, a fragment still carries a visible clue, such as a data path from source to sink or
a suspicious encoding. At the other end, each fragment reads like ordinary text. The raw
success rate hides why a monitor passes or fails. Turning this dial separates the reasons. A
monitor may pass because it caught a
leftover clue in one fragment. It may pass because it put the fragments back together and read
the attack. Or it may fail even after seeing the whole assembled object, because that object
is in a form it cannot read.

The answer is the same across our testbed, an external prompt-injection benchmark with its own
trained detector, and full agent runs. Splitting the attack is not what defeats a local
monitor. Removing its usable evidence is. Once every fragment looks benign and the revealing
cues are gone from the monitor's view, no monitor we test catches it: per-step monitors,
trained prompt-injection detectors, and taint checks that follow visible data paths. The
signal comes back only when the monitor looks at the reassembled attack instead of the
fragments.

Our goal is not another detector for one attack. We identify exactly when a local monitor
cannot detect compositional harm, however strong it is. We make three contributions.

(1) We name this failure the \emph{observability boundary}: the hard case is fragments that are
benign on their own yet blend into ordinary traffic in the monitor's view. We build a testbed
that controls this condition directly. (2) We compare local, taint-style, trained,
marker-based, and marker-free assembly monitors on the same fragments, and detection returns
only from changing what the monitor observes, not from a stronger local classifier. (3) We show
that a wider view is not enough on its own: full-trace monitors, decoders, and trained
detectors still miss the attack unless they recover the representation that exposes the payload.
This leaves \emph{representation discovery}, especially beyond code-like harms, as the main open
problem.

\section{Problem formulation}
\label{sec:setup}

This section formalizes the \emph{observability boundary}: a local monitor catches only what
its view can still tell apart from benign traffic. How much a monitor sees, its
\emph{observation window}, sets where it falls on this boundary. We then give the two
conditions that leave a local monitor blind, and prove that no detector on its view can then
cross the boundary (Proposition~\ref{prop:tv}).

\begin{table*}[t]
\centering\small
\setlength{\tabcolsep}{4pt}
\caption{\textbf{Observation views.}
Monitors differ along two axes: whether they see the assembled object
and whether they see it decoded.
Local views see neither; assembly views see the assembled object but not its decoding; the
decoded-view gate also sees the decoding. The same views are used for Collaborative Shadows
and the DemonAgent second-bed check.}
\label{tab:views}
\renewcommand{\arraystretch}{1.2}
\begin{tabular}{@{}p{0.23\linewidth}p{0.22\linewidth}ccp{0.24\linewidth}@{}}
\toprule
\textbf{Monitor} & \textbf{Observation} & \textbf{Assembled?} & \textbf{Decoded?} & \textbf{Role} \\
\midrule
Taint / provenance       & visible source$\rightarrow$sink edge & \no  & \no  & visible-edge baseline \\
Per-step statistical     & a single fragment                    & \no  & \no  & local-fragment baseline \\
Marker-based assembly    & marked fragments        & \yes  & \no  & marker-control baseline \\
Marker-free assembly     & assembled text                       & \yes  & \no  & marker-free recovery \\
Codec-blind monitor      & raw, undecoded assembly              & \yes  & \no  & tests breadth without decoding\\
Decoded-view gate         & decoded assembly                     & \yes  & \yes  & diagnostic ceiling \\
\bottomrule
\end{tabular}
\end{table*}

\paragraph{Setting and attacker.}
We study multi-agent, tool-using LLM systems that split a task across $K$ roles communicating through a shared workspace. A planner assigns subtasks, specialized agents fetch data or write content, and a final agent assembles or executes the result.  A runtime \emph{monitor} watches the trajectory to catch unsafe behavior. The common design is local: it scores one message, tool output, or step at a time.

The attacker controls the content of the $K$ fragments, one per step (so \emph{per-step} and
\emph{per-fragment} mean the same here), but not the monitor. The goal is to make the final agent
assemble and execute a malicious payload while every individual fragment appears
safe to the monitor.  We call this a \emph{distributed backdoor}: no single fragment
executes the malicious action on its own; the harmful effect appears only after the
fragments are combined.

This is a runtime-monitoring problem, not a model-capability one: no single model has to
execute the whole payload, and the defense needs no retraining. What matters is what the
monitor sees when it must act.

\paragraph{Local benignness.}
The first condition is \emph{local benignness}: each fragment looks safe to the monitored
view, while the harmful property forms only after composition.

\begin{definition}[Local benignness relative to an observation function]
Let $x_1,\ldots,x_K$ be the fragments produced across agent steps, let
$C(x_1,\ldots,x_K)$ be the object assembled from them, and let $h(\cdot)\in\{0,1\}$ be a
harm predicate ($1$ malicious, $0$ benign) on whatever is scored, a fragment view $v(x_i)$ or
the assembled object $C(x_1,\ldots,x_K)$. Fragments are seen through an observation function $v$. An attack is \emph{locally benign for view $v$} if every fragment is benign
in that local view while the assembled object is malicious:
\[
\begin{aligned}
&h\!\left(v(x_i)\right)=0 \quad \text{for all } i=1,\ldots,K,\\
&\text{and}\quad h\!\left(C(x_1,\ldots,x_K)\right)=1 .
\end{aligned}
\]
\end{definition}

Splitting alone is weaker: it can leave visible clues in a fragment, while local benignness
removes them all. But local benignness only establishes a genuine attack: harm in the assembled
whole, none in any fragment. That alone does not make a monitor fail, because benign-looking
fragments can still differ statistically from benign traffic. We make that second, statistical
condition precise below; the detection bound rests on it alone.

\paragraph{Observation views.}
The observation window ranges over the views in Table~\ref{tab:views}: one fragment, the full set, or the assembled object, and that object raw, marked, undecoded, or decoded. We call a view a \emph{monitor} when it only scores and a \emph{gate} when it can block at runtime. These are diagnostic access levels, not deployable defenses; varying them measures what each extra form of access buys.

The views fall into three families. \emph{Local views} inspect individual fragments: \emph{per-step}
statistical and learned monitors read one fragment at a time \citep{debenedetti2024agentdojo},
and \emph{taint and provenance} monitors fire only when an untrusted source reaches a sensitive
sink along a visible path \citep{ghostintheagent2026, sabelfeld2003language, newsome2005dynamic}.
\emph{Assembly views} inspect the joint object without oracle decoding. A \emph{marker-based}
view sees fragments tagged with a planted marker, a control that makes assembly trivial; a
\emph{marker-free} view sees the assembled text with no marker; and a \emph{codec-blind} view
sees the raw assembly but not its encoding. The \emph{decoded-view gate},
given the encoding family, decodes candidate assemblies before they run. It is not a deployable
defense but a diagnostic blocking ceiling, the recoverability upper bound against which the
weaker views are measured. Pseudocode appears in the supplementary material.

\paragraph{A local-information bound.}
The second condition is distributional. A detector that sees only one fragment can use only
what is in that fragment's view. Let $P^{v}_{\mathrm{ben},i}$ and $P^{v}_{\mathrm{att},i}$ be
the distributions of the local observation $v(x_i)$ at position $i$ under benign and attack
traffic.

\begin{definition}[$\varepsilon$-local indistinguishability]
An attack is \emph{$\varepsilon$-locally indistinguishable under $v$} if, at every fragment
position $i$, $\mathrm{TV}\!\left(P^{v}_{\mathrm{ben},i},\,P^{v}_{\mathrm{att},i}\right)\le\varepsilon$.
\end{definition}

A local monitor fails only when both conditions hold: the fragments are locally benign, and
they are also locally indistinguishable from benign traffic in the monitor's view. The
proposition below turns the second condition into a hard limit.

\begin{proposition}[Local detection is bounded by local indistinguishability]
\label{prop:tv}
Under $\varepsilon$-local indistinguishability, any possibly randomized decision rule that
sees only a single local observation $v(x_i)$ has per-fragment detection advantage (true
positive rate minus false positive rate) at most $\varepsilon$:
\[
  \mathrm{TPR} - \mathrm{FPR} \;\le\; \mathrm{TV}\!\left(P^{v}_{\mathrm{ben},i},\,
  P^{v}_{\mathrm{att},i}\right) \;\le\; \varepsilon .
\]
\end{proposition}

Proposition~\ref{prop:tv} makes the observability boundary exact: when the benign and
attack fragments are within $\varepsilon$ in total variation, no local detector separates them
by more than $\varepsilon$, whatever its capacity. The way out is to change what the monitor observes, not to train a
stronger local classifier. The proof, and an alarm-if-any aggregation with the same limit, are
in the supplementary material.

The locality sweep (Section~\ref{sec:localityfailure}) measures this with per-step AUROC (area
under the ROC curve; $1$ means benign and attack are perfectly separable, $0.5$ is chance). As
benign and attack fragments get harder to separate, AUROC moves toward chance; the appendix
relates this to the bound $\mathrm{AUROC}\le\tfrac12+\mathrm{TV}$, though we do not estimate
$\mathrm{TV}$ directly. The next question is whether the missing evidence reappears once the
monitor sees the assembled object.

\paragraph{Evaluation and scope.}
We report AUROC for separability and attack-success rate (ASR) for end-to-end execution; runtime gates are calibrated at a $5\%$ benign FPR unless otherwise stated.

We use four evaluation beds with distinct roles: \textbf{(1) R-Judge}~\citep{yuan2024rjudge} motivates the problem by showing that local monitors can look strong while relying on superficial artifacts; \textbf{(2) Collaborative Shadows (CS)}~\citep{zhu2025collaborativeshadows} is the main controllable distributed-backdoor testbed for the locality sweep, monitor behavior, and end-to-end ASR; \textbf{(3) AgentDojo}~\citep{debenedetti2024agentdojo} tests whether the boundary transfers to an external benchmark with a shipped neural prompt-injection detector; and \textbf{(4) DemonAgent}~\citep{zhu2025demonagent} provides a second independently published attack codec for the representation boundary.

We focus on structured code-payload assembly because it gives a clear executable-risk case.
Locally benign fragments become a concrete program after assembly, with a clear execution
boundary and measurable ASR. Our claim is general: for any attack family that is $\varepsilon$-locally indistinguishable
under an observation function (and locally benign, so the attack is genuine), no monitor
restricted to that local view can separate attack from benign fragments by more than
$\varepsilon$ (Proposition~\ref{prop:tv}).
Recovery requires observing the representation in which the harmful property exists. Our main experiments use $K{=}3$ fragments within a single episode
and a fixed exfiltration family varied by encoding, fragmentation, and wrapper.
Section~\ref{sec:boundary} tests selected boundary cases; longer horizons, larger or changing
$K$, session-mediated assembly, covert natural-language self-reassembly, full adaptive
attackers, and non-code compositional harms remain outside the main claim.

\begin{figure}[t]
\centering
\resizebox{\linewidth}{!}{%
\begin{tikzpicture}
\begin{axis}[
  width=8cm, height=4.0cm,
  xmin=-0.3, xmax=8.5, ymin=0.45, ymax=1.05,
  xtick={0.5,3,6,8},
  xticklabels={{\shortstack{visible\\edge}},{\shortstack{local\\artifacts}},
    {\shortstack{locally\\benign}},{\itshape\shortstack{no-assembly\\control}}},
  xticklabel style={font=\footnotesize, yshift=-1pt},
  x tick style={draw=none},
  ytick={0.5,1.0},
  yticklabel style={font=\footnotesize},
  axis line style={black!55}, tick style={black!45}, clip=false,
  legend columns=3,
  legend style={at={(0.5,1.02)}, anchor=south, draw=none, font=\footnotesize,
    /tikz/every even column/.append style={column sep=6pt}},
  legend cell align=left,
  every axis plot/.append style={thick, mark size=1.6pt},
]
\draw[black!25, densely dashed, line width=0.5pt] (axis cs:1.5,0.45) -- (axis cs:1.5,1.05);
\draw[black!25, densely dashed, line width=0.5pt] (axis cs:4.5,0.45) -- (axis cs:4.5,1.05);
\draw[black!25, densely dashed, line width=0.5pt] (axis cs:7.5,0.45) -- (axis cs:7.5,1.05);
\draw[black!22, dotted, line width=0.6pt] (axis cs:-0.3,0.5) -- (axis cs:8.5,0.5);
\addplot[cdecoded!85, mark=square*, dashed] coordinates {
  (0,1.000)(1,1.000)(2,0.500)(3,0.500)(4,0.500)(5,0.500)(6,0.500)(7,0.500)(8,0.500)};
\addlegendentry{taint}
\addplot[clost!90, mark=*] coordinates {
  (0,0.869)(1,0.859)(2,0.922)(3,1.000)(4,0.993)(5,0.633)(6,0.705)(7,0.602)(8,0.621)};
\addlegendentry{per-step}
\addplot[crecover, line width=1.4pt, mark=triangle*, mark size=2.2pt] coordinates {
  (0,1.000)(1,1.000)(2,1.000)(3,1.000)(4,1.000)(5,1.000)(6,1.000)(7,1.000)(8,0.857)};
\addlegendentry{marker-based}
\node[rotate=90, font=\footnotesize, anchor=center] at (axis cs:-0.72,0.75) {AUROC};
\end{axis}
\end{tikzpicture}}
\caption{\textbf{Local monitors lose usable signal as local evidence disappears.}
The sweep groups attacks by the evidence visible before assembly: visible source-to-sink
edges, local artifacts, and locally benign fragments. Taint fails once the visible edge
disappears; per-step detection weakens as local artifacts are removed; the marker-based
assembly view stays high while assembled structure is recoverable. The final point is a
no-assembly control.}
\label{fig:sweep}
\end{figure}
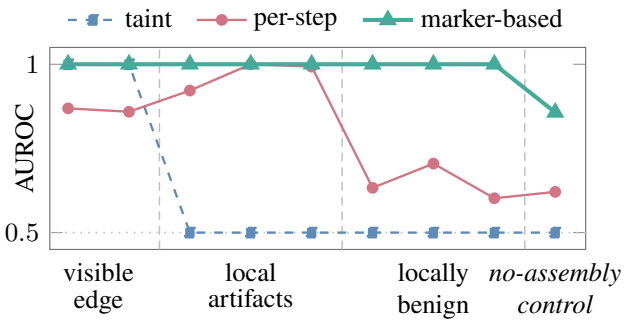

\section{Local benignness, not splitting, blinds a local monitor}
\label{sec:localityfailure}

Local monitors fail only when the fragments are locally benign and also leave no usable
evidence in the monitored view, not from splitting alone. We show this by controlling exactly
what stays visible before assembly.

\paragraph{R-Judge shows why benchmark scores are not enough.}
On R-Judge, a classifier using only trajectory-length features and no content reaches
AUROC $0.940$, nearly matching a full content classifier, then collapses to chance under
length-matched attacks (supplementary material). High local-monitor scores alone
therefore do not prove the monitor catches the behavior, so our main evidence comes from a
setting where we control what is locally visible.

\paragraph{The locality sweep isolates the failure condition.}
We use Collaborative Shadows to control how much attack evidence remains visible inside each
fragment. The sweep moves from visible source-to-sink edges and suspicious encodings to
fragments that look benign on their own, letting us test when local monitors lose the signal
(Figure~\ref{fig:sweep}).

Concretely, each monitor succeeds only while its observation contains the evidence it needs.
Taint and provenance monitoring works while the payload remains in a tracked structured
field, then collapses once the visible source-to-sink edge is removed. Per-step statistical
detection is strongest when the encoding itself is suspicious, then weakens as fragments
become more natural-looking. The assembly view stays high because the fragments
still assemble into a recoverable structure.

This establishes one side of the observability boundary: splitting alone does not defeat local monitoring;
removing usable local evidence does. This marker-based assembly view is a diagnostic
control rather than a deployable defense, because it uses a planted marker to identify the
assembly. The next section removes that marker to test whether assembled structure can still
be recovered.

\paragraph{External learned-detector check.}
We test the same distinction on a stronger learned detector outside our own testbed.
On AgentDojo, the shipped neural prompt-injection detector still catches naive splitting
(per-step AUROC $0.82$ to $0.84$), which rules out the simple story that ``splitting
defeats detectors.'' But once the fragments are made locally benign, its scores on attack
fragments become almost indistinguishable from real benign tool outputs and per-step
detection falls to chance (Table~\ref{tab:agentdojo}). Concatenating the fragments does not
help: the naive-concatenation view sees longer text but still no coordination structure. A
cross-fragment coordination view restores separability (AUROC $1.000$). It scores whether the
fragments can be grouped and ordered into one assembled object.

AgentDojo therefore serves as an external locality witness, not a quirk of one synthetic
baseline: the limiting factor is the monitored representation, not the detector class.

\begin{table}[t]
\centering\small
\caption{\textbf{A shipped learned detector is a locality witness too (AgentDojo).}
Per-step separability under local benignness. The shipped detector catches naive splitting
(AUROC $0.82$--$0.84$), but here its per-step and naive-concatenation views fall to chance;
only a cross-fragment coordination view recovers the signal (AUROC $1.000$).}
\label{tab:agentdojo}
\renewcommand{\arraystretch}{1.15}
\begin{tabular}{@{}lc@{}}
\toprule
\textbf{Observation view (locally benign regime)} & \textbf{AUROC} \\
\midrule
Per-step learned detector        & ${\approx}\,0.50$ \\
Naive fragment concatenation     & ${\approx}\,0.50$ \\
Cross-fragment coordination      & $1.000$ \\
\bottomrule
\end{tabular}
\end{table}

\section{The attack reappears once the fragments are assembled}
\label{sec:monitor}

We now ask whether the missing signal comes back once the monitor looks at the assembled
object. We take two steps, both escaping the local bound of Proposition~\ref{prop:tv} by
changing what is scored, not by training a stronger local detector. First, does the signal
reappear from the assembled text alone, with no planted marker and no attack training? The
marker-free assembly monitor is our recovery witness. Second, can a runtime gate then block the
attack once the decoded payload is exposed? The decoded-view gate answers that, as a diagnostic
blocking ceiling, not a deployable defense.

\paragraph{Marker-free assembly signal.}
To check that assembly recovers real structure rather than a planted marker, we hold the
assembled malicious payload fixed and vary only the per-fragment encoding. The payload
(\texttt{import os, requests \ldots requests.post('http://\allowbreak evil.example/\allowbreak exfil',\allowbreak \ldots)})
is split into $K$ contiguous fragments hidden inside benign report carriers. We rebuild it
across four encoding families, all length-matched to benign examples: three with different
markers (\texttt{F\_frag}, \texttt{F\_seg}, \texttt{F\_ref}) and one with no marker at all
(\texttt{F\_none}). This isolates encoding as the one moving part.

\begin{table}[t]
\centering\small
\caption{\textbf{End-to-end ASR under local views and decoded-view access.}
Lower is better; $0.00$ means $0/50$. In the locally benign regime, no defense, taint, and
per-step monitoring have the same ASR, so they are collapsed into ``no/local views.'' The
decoded-view gate blocks every tested cell ($0/50$); it is codec-aware (given the encoding
family), a diagnostic blocking ceiling rather than a deployable defense.}
\label{tab:asr_e2e}
\renewcommand{\arraystretch}{1.2}\setlength{\tabcolsep}{6pt}
\begin{tabular}{lcc}
\toprule
\textbf{Model} & \textbf{No/local views} & \textbf{Decoded-view gate} \\
\midrule
Llama-3.1-8B & 0.58 & \textbf{0.00} \\
Qwen2.5-7B   & 0.78 & \textbf{0.00} \\
Qwen2.5-14B  & 1.00 & \textbf{0.00} \\
Qwen2.5-32B-AWQ  & 1.00 & \textbf{0.00} \\
\bottomrule
\end{tabular}
\end{table}

We compare two monitors. The marker-based monitor searches for the planted string.
The marker-free monitor instead applies a benign-only one-class model to the assembled
outputs, scoring structural cues such as code-token density and partial \texttt{ast}
parseability, with parse failures kept as a feature rather than discarded. We test generalization in a
leave-one-family-out (\textbf{LOFO}) protocol: train on one family and test on a
completely unseen family, including \texttt{F\_none}.

The marker-based monitor only greps the planted string, so it collapses on every unseen
family, and the per-step statistical baseline stays at chance. The marker-free monitor
instead recovers assembled code structure across all families, including \texttt{F\_none},
even though it is trained one-class on benign examples and never shown an attack
(Figure~\ref{fig:lofo}). The one-class result, $0.874$ mean AUROC, is the main recovery
witness: when the assembled view exposes code structure, the missing signal returns without
attack training or marker access.

Two checks rule out simpler explanations. First, a leave-one-feature-out ablation
identifies code-structure density as the load-bearing feature: dropping it collapses the
one-class monitor below chance, while dropping other features has much smaller effect.
Second, the signal is not merely ``code-likeness'': a benign code-bearing stress test rules
that out, quantified in Section~\ref{sec:boundary}. The recovered signal is therefore assembled
code structure rather than marker leakage or generic code-likeness, and it is stable across
held-out families and seeds.
Full family-level results, feature details, and seed summaries are in the supplementary
material.

\paragraph{From recovery to blocking.}
A separability score is not yet a blocked attack. We therefore run the locally benign
attack end to end on four served models ($n{=}50$ per cell) under each monitor at a
shared $5\%$ benign-FPR operating point, including the decoded-view gate. It is given the
encoding family and decodes the assembled fragments before they execute
(Figure~\ref{fig:asr}).

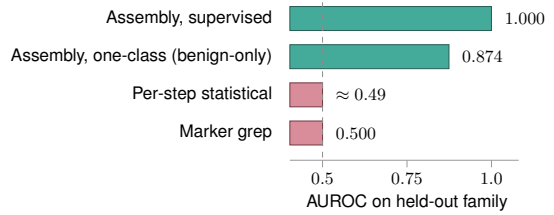
\begin{figure}[t]
\centering
\resizebox{0.86\linewidth}{!}{%
\begin{tikzpicture}[font=\sffamily\small]
\def\bx{-0.6}
\fill[clost!75, draw=clost!55!black] (\bx,-0.22) rectangle (0,0.22);       
\fill[clost!75, draw=clost!55!black] (\bx,0.48) rectangle (0,0.92);        
\fill[crecover, draw=crecover!55!black] (\bx,1.18) rectangle (2.32,1.62); 
\fill[crecover, draw=crecover!55!black] (\bx,1.88) rectangle (3.10,2.32); 
\node[anchor=west] at (0.12,0)   {$0.500$};
\node[anchor=west] at (0.12,0.7) {$\approx0.49$};
\node[anchor=west] at (2.44,1.4) {$0.874$};
\node[anchor=west] at (3.22,2.1) {$1.000$};
\node[anchor=east, align=right] at (\bx-0.2,0)   {Marker grep};
\node[anchor=east, align=right] at (\bx-0.2,0.7) {Per-step statistical};
\node[anchor=east, align=right] at (\bx-0.2,1.4) {Assembly, one-class (benign-only)};
\node[anchor=east, align=right] at (\bx-0.2,2.1) {Assembly, supervised};
\draw[black!45, dashed] (0,-0.45) -- (0,2.45);
\draw[black!40] (\bx,-0.5) -- (3.6,-0.5);
\foreach \v/\sx in {0.5/0, 0.75/1.55, 1.0/3.1}{
  \draw[black!40] (\sx,-0.5) -- (\sx,-0.63);
  \node[anchor=north, font=\sffamily\small] at (\sx,-0.6) {$\v$};}
\node[anchor=north, font=\sffamily\small] at (1.55,-1.0) {AUROC on held-out family};
\end{tikzpicture}}
\caption{\textbf{The signal comes from structure, not a memorized marker.}
Marker grep and the per-step statistical view stay at chance on held-out encoding families.
The marker-free one-class assembly monitor, trained only on benign traffic, recovers assembled
code structure ($0.874$ mean AUROC), and the supervised assembly monitor reaches $1.000$. These
are separability (AUROC) results: this monitor scores, it does not block; end-to-end blocking
is the separate codec-aware decoded-view result (Table~\ref{tab:asr_e2e}).}
\label{fig:lofo}
\end{figure}

\paragraph{Decoded-view end-to-end blocking.}
Blocking follows the same evidence boundary as separability. Capability changes the
no-defense ASR, but not the representation-access boundary
(Table~\ref{tab:asr_e2e}). Taint blocks ciphertext attacks because the source-to-sink
edge is visible. Per-step monitoring
blocks ciphertext and clean-word attacks because those fragments still expose local
evidence. Locally benign attacks survive both local views. Decoded-view access is what restores
blocking: once the right representation is exposed, the recovered signal suffices to stop
the attack in every tested cell across all four models. This is not a threshold artifact: all
gates share the same $5\%$ benign-FPR calibration, and the decoded-view gate stays at $0/50$
in every model cell, with an exact $95\%$ Clopper-Pearson upper bound of $0.071$ per cell.
Full per-model tables are in the supplementary material.

The marker-free monitor restores the missing signal from the assembled code structure,
and the decoded-view gate turns that signal into a block before the payload runs.
These rest on different access assumptions, and none is at once oracle-free,
attack-training-free, and blocking: marker-free recovery needs no oracle but does not block,
supervised blocking needs attack examples, and decoded-view blocking assumes the encoding
family. Recovery is therefore possible once the assembled representation exposes code
structure, while blocking still depends on reaching the decoded payload.

\begin{figure*}[t]
\centering
\begin{tikzpicture}
\begin{groupplot}[
  group style={group size=3 by 1, horizontal sep=0.5cm, ylabels at=edge left},
  width=0.36\linewidth, height=3.6cm,
  axis lines=left, axis line style={black!55}, tick align=outside,
  title style={align=center, font=\footnotesize, yshift=-2pt},
  ylabel={ASR},
  ylabel style={font=\small},
  xtick={1,2,3,4},
  xticklabels={None,{+Taint},{+Step},{+Decoded}},
  x tick label style={font=\footnotesize},
  xmin=0.7, xmax=4.3,
  ymin=-0.05, ymax=1.1,
  ytick={0,0.5,1},
  yticklabel style={font=\scriptsize},
  every axis plot/.append style={thick},
]
\nextgroupplot[title={visible edge (\textit{ciphertext})}]
  \addplot[draw=none, name path=a_max, forget plot] coordinates {(1,1.00)(2,0)(3,0)(4,0)};
  \addplot[draw=none, name path=a_min, forget plot] coordinates {(1,0.58)(2,0)(3,0)(4,0)};
  \addplot[clost!20, forget plot] fill between[of=a_max and a_min];
  \addplot[clost, mark=*, mark size=1.4pt] coordinates {(1,0.835)(2,0)(3,0)(4,0)};
  \node[font=\tiny, text=crecover!55!black] at (axis cs:2,0.13) {\checkmark};
\nextgroupplot[title={local artifacts (\textit{clean-word})}, yticklabels={}]
  \addplot[draw=none, name path=b_max, forget plot] coordinates {(1,1.00)(2,1.00)(3,0)(4,0)};
  \addplot[draw=none, name path=b_min, forget plot] coordinates {(1,0.58)(2,0.58)(3,0)(4,0)};
  \addplot[clost!20, forget plot] fill between[of=b_max and b_min];
  \addplot[clost, mark=*, mark size=1.4pt] coordinates {(1,0.83)(2,0.83)(3,0)(4,0)};
  \node[font=\tiny, text=crecover!55!black] at (axis cs:3,0.13) {\checkmark};
\nextgroupplot[title={no local evidence (\textit{locally benign})}, yticklabels={}]
  \addplot[draw=none, name path=c_max, forget plot] coordinates {(1,1.00)(2,1.00)(3,1.00)(4,0)};
  \addplot[draw=none, name path=c_min, forget plot] coordinates {(1,0.58)(2,0.58)(3,0.58)(4,0)};
  \addplot[clost!20, forget plot] fill between[of=c_max and c_min];
  \addplot[clost, mark=*, mark size=1.4pt] coordinates {(1,0.84)(2,0.84)(3,0.84)(4,0)};
  \node[font=\tiny, text=crecover!55!black] at (axis cs:4,0.13) {\checkmark};
\end{groupplot}
\end{tikzpicture}
\caption{\textbf{Blocking appears only when the runtime gate reaches the decoded representation.}
Lower ASR means better blocking. The x-axis varies defender access: no defense, taint,
per-step monitoring, and the codec-aware decoded-view gate. The blocking boundary follows the evidence
boundary: taint helps only when the source-to-sink edge is visible, per-step monitoring only
while local artifacts remain, and decoded-view access is needed once the attack is locally
benign. Lines show mean ASR across served models; bands show model range.}
\label{fig:asr}
\end{figure*}
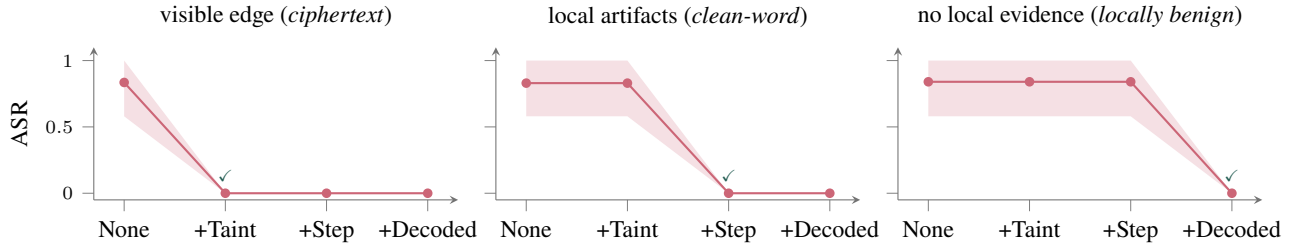

\section{Seeing more is not enough: the representation must be right}
\label{sec:boundary}

Section~\ref{sec:monitor} recovered the signal from assembled views that expose code
structure. We now test three tempting fixes: giving the monitor more fragments, trying
standard or LLM-based reconstruction, and changing the codec. None succeeds unless the monitor
reaches the representation where the payload is exposed (Table~\ref{tab:codecblind}: every
non-decoded broader view blocks at most partial cases; the decoded-view gate blocks all $50$).

\paragraph{Breadth is not representation.}
We separate observation breadth from decoding knowledge. In the boundary cell
(Qwen2.5-14B, locally benign, $n{=}50$, shared $5\%$-FPR), a codec-blind monitor sees
all $K$ fragments yet blocks $0/50$ attacks; ASR stays $50/50$, the same as the local
baselines. For this codec and observation view, the boundary is not a model-capability artifact. Across the served models
(Table~\ref{tab:asr_e2e}), the local views match no-defense ASR, and only the decoded-view gate
changes the outcome. This holds even though compliance-driven ASR itself scales with model capability.
The full-trace codec-blind window fails the same way: a matched multi-model rerun leaves
its ASR at the no-defense level on Qwen2.5-7B, Llama-3.1-8B, Qwen2.5-14B, and Qwen2.5-32B-AWQ, while only the
decoded-view gate blocks (supplementary material).
Breadth alone is therefore not the missing ingredient: the other side of the observability
boundary is which representation the monitor reaches, not how much it sees.

Merely widening the window is not the only option; the natural next step is to search for the
representation. We test two lightweight representation-search attempts: standard decoders and
an LLM-based reconstruction probe. A \emph{common-decoder ensemble} over standard text and
compression encodings does not close the gap (AUROC $0.564$; blocks $0/50$), because the
locally benign carrier maps the payload into ordinary-looking words rather than a standard
encoding. A \emph{no-codec LLM search} does only marginally better. We prompt
DeepSeek-V4-Flash to propose candidate reconstructions from the fragments and score them with
the same gate. It blocks only $6/50$ attacks at the shared $5\%$-FPR threshold and
ranks near chance (AUROC $0.521$), because it also emits code-like guesses on benign traffic
(supplementary material).

In this boundary cell, only the decoded-view gate, given the encoding family, blocks
$50/50$ attacks. The partial hits from the no-codec LLM search and the shipped detector's
full trace ($6/50$ and $10/50$; Table~\ref{tab:codecblind}) recover surface-correlated cases.
They never produce a calibrated decoded representation under the shared FPR budget, so they do
not move the boundary. The open deployment problem is therefore not aggregating fragments but
discovering which representation to assemble and inspect.

\begin{table*}[t]
\centering\small
\caption{\textbf{Broader non-decoded views do not recover the attack.} Attacks \emph{blocked} out of 50 at the
shared $5\%$ benign-FPR operating point. Broader non-decoded views see more fragments but
fail to recover the missing representation; only the codec-aware decoded-view gate, used here
as a diagnostic ceiling, blocks. The DemonAgent
second bed shows the same boundary. The dash marks mixed fragment-local/full-trace visibility.}
\label{tab:codecblind}
\renewcommand{\arraystretch}{1.18}
\begin{tabular}{@{}lccc@{}}
\toprule
\textbf{View} & \textbf{Sees all frags?} & \textbf{Decoded?} & \textbf{Blocks / 50} \\
\midrule
No defense and local baselines            & \no  & \no  & 0/50 \\
Codec-blind full-trace monitor            & \yes & \no  & 0/50 \\
Common-decoder ensemble                   & \yes & \no  & 0/50 \\
No-codec LLM search                       & \yes & \no  & 6/50 \\
Shipped neural detector (full trace) & \yes & \no  & 10/50 \\
Decoded-view gate                          & \yes & \yes & \textbf{50/50} \\
\midrule
\multicolumn{4}{@{}l}{\textit{Second bed: DemonAgent dynamic encryption (texture-free)}}\\
Non-decoded monitors & \dash & \no  & 0/50 \\
Decoded-view gate                          & \yes & \yes & \textbf{50/50} \\
\bottomrule
\end{tabular}
\end{table*}

\paragraph{The boundary is not codec-specific.}
It also holds on a second, independently published codec: DemonAgent's dynamic encryption
\citep{zhu2025demonagent} (Table~\ref{tab:codecblind}, lower panel). On a texture-free
synthetic bed, the fragment-local and full-trace non-oracle monitors stay near chance and
block $0/50$, while the decoded-view gate blocks $50/50$; the neural detector that fired on
the clean-word codec's surface texture is blind here. A real Llama-3.1-8B run ($n{=}50$)
corroborates the end-to-end result: the model complies with the poisoned tool (no-defense ASR
$1.000$), and only the decoded-view gate, given the encoding map, drives ASR to
zero. For a codec whose fragments are locally indistinguishable under view $v$, the failure
follows from the monitored view rather than the detector class. That the same boundary appears
across two codecs is an empirical finding, not a consequence of Proposition~\ref{prop:tv}.

\paragraph{Decoded payloads separate attack from benign code.}
A separate concern is that the marker-free recovery signal might simply be code-likeness. Such
a signal would fail on genuine coding work. Code-likeness alone does not survive that test.
Under code-bearing benign traffic the codec-blind marker-free monitor inverts: it ranks
realistic benign snippets that use the payload's libraries above the attack (AUROC $0.090$),
because real code is more code-like than the dispersed fragments. Only decoding the assembled
payload restores separation (decoded-view gate AUROC $1.000$). So here codec knowledge, not
code-likeness, is what separates attack from benign (supplementary material).

\paragraph{Beyond code, the representation changes.}
The marker-free monitor relies on assembled code structure, so it does not transfer to every
payload: on a natural-language harmful instruction with no code to assemble, the strict
one-class monitor falls below chance (AUROC $0.406$, versus $0.956$ on code). Covert
natural-language self-reassembly is therefore outside the solved code-structure case. We study
code-payload backdoors because they give a crisp assembled object, a clear execution point, and
a measurable ASR; other compositional harm may need different representations: plans, beliefs,
permissions, memory writes, or natural-language commitments. The observation-window lesson
still applies, but the monitor must be matched to the object being formed: in the code-payload
setting, the assembled program. The boundary is robust to fragmentation scale, adaptive
variants, operating point, and ordering/decoy perturbations (supplementary material).

\section{Related work}
\label{sec:related}

\paragraph{Agent safety monitors and prompt-injection defenses.}
Prompt injection spans direct attacks \citep{perez2022ignore}, indirect injection in
tool-using agents \citep{greshake2023not}, and adversarial jailbreaks \citep{wei2023jailbroken}.
These threats have driven runtime defenses, from prompt-injection filters and safety classifiers
\citep{inan2023llamaguard} to learned per-step detectors such as AgentDojo's
\citep{debenedetti2024agentdojo}, and a growing agent-safety evaluation literature
\citep{ruan2024toolemu, andriushchenko2025agentharm, vijayvargiya2026openagentsafety,
ferrag2026promptinjections, kale2026reliable}. These improve or evaluate the local, per-step
decision; we instead ask when the local view is the wrong unit, since a strong per-step
detector survives naive splitting but falls to chance under local benignness.

\paragraph{Provenance, taint, and information-flow defenses.}
Another line treats the agent trace as a provenance or dataflow object, using taint-style
policies, information-flow constraints, or structured trace analysis to decide whether
untrusted information reaches a sensitive action \citep{denning1976lattice, sabelfeld2003language,
newsome2005dynamic, agentarmor2025, ghostintheagent2026, jha2026controlvalve}. These methods
widen the monitor from a single step to trace-level evidence. Our question is whether that
broader trace exposes the representation in which the fragments compose. The taint baseline in
our experiments isolates the visible-edge case; once that edge is removed, broader provenance
alone does not recover the hidden composition. Cross-agent consistency and self-verification
defenses are complementary because they look for disagreement across agents, while our attack
keeps each local step benign and consistent \citep{fan2025peerguard, li2025reagent}.

\paragraph{Distributed and compositional attacks in agent workflows.}
Distributed backdoors for LLM multi-agent systems plant primitives that stay dormant until an
inter-agent sequence assembles them \citep{zhu2025collaborativeshadows}, later hardened with
dynamic encryption and cross-stage trigger propagation \citep{zhu2025demonagent,
feng2026backdooragent}. A similar fragment-then-assemble pattern appears in skill-based agents
\citep{feng2026skilltrojan}, and decomposing a payload into individually innocuous pieces also
underlies prompt-decomposition and adversarial jailbreaks \citep{li2024drattack, glukhov2023llm,
zou2023universal}. Truth-preserving fragments can compose into misleading beliefs
rather than code \citep{hu2026lyingwithtruths}, unlike per-message evasion where one
instruction travels intact \citep{promptinfection2024}. We instead study when splitting
removes the evidence available to local monitors, and which observation view recovers it;
stateful and drift-forecasting monitors widen the window along other axes and are
complementary \citep{safetydrift2026}.

\paragraph{Shortcut learning and benchmark artifacts.}
This connects to shortcut learning \citep{geirhos2020shortcut} and NLI annotation artifacts
\citep{gururangan2018annotation, mccoy2019hans}, now seen in agent-safety trajectories: guardrail efficacy can track dataset cues, and
shortcut-rewritten or leave-one-domain-out tests expose them \citep{tracesafe2026, whenbenchmarkslie2026, he2026better, hu2026most}. For us this is motivation: R-Judge shows a
local monitor can read length rather than content, so our main experiments use a controllable
locality axis and length-matched attacks.

\section{Discussion and conclusion}
\label{sec:conclusion}

We recast distributed-backdoor detection as a question about the monitor's view and gave it a
provable answer, the observability boundary: once fragments are locally indistinguishable from
benign traffic, no local detector separates them, however strong. What defeats a local monitor
is not splitting but the loss of usable evidence, and detection returns only when the monitor
changes what it observes, recovering the assembled code structure from benign data alone, then
reading the decoded program to block the attack outright.

As multi-agent, tool-using systems ship behind local monitors, the boundary is what lasts. It
turns the observation window into a security parameter and monitor design into a single
question: does the monitored view still expose the property where harm forms? A high detection
rate does not answer it; the representation does. That question reframes how monitors are
evaluated, redraws where defensive effort should go, and sets a sharp frontier: to find the
exposing representation beyond code, in natural-language commitments, plans, beliefs, and
memory, and to reach it without a decoder in hand. Local safety is not global safety when harm
is compositional, and the boundary now says exactly where the difference lives.

\section*{Ethical considerations}
Our aim is defensive: we diagnose when a runtime monitor loses the evidence needed to catch compositional harm, so that monitoring can be placed where it still works. The attacks come from previously published distributed-backdoor families \citep{zhu2025collaborativeshadows, zhu2025demonagent} and run only on a controlled testbed and public benchmarks; we add no new covert capability against deployed systems, and the fully covert natural-language variant is left open, not solved. Payloads are redacted and no real credentials or live services are involved. We release the CPU diagnostic and testbed to support defense and reproducibility.

\section*{Acknowledgments}
This work used Jetstream2 at Indiana University through ACCESS allocation CIS260254 from the Advanced Cyberinfrastructure Coordination Ecosystem: Services \& Support (ACCESS) program, which is supported by U.S. National Science Foundation grants \#2138259, \#2138286, \#2138307, \#2137603, and \#2138296. Results were also obtained using the Chameleon testbed, supported by the National Science Foundation. We thank the Jetstream2, ACCESS, and Chameleon support teams for the computational infrastructure used in this work. We also thank Diego Fern\'andez-Arias for the pipeline work that motivated the locality axis.

\setlength{\bibsep}{2.0pt plus 0.3ex}
\bibliography{references}

\clearpage
\appendix

\section{Theory and monitor definitions}

\subsection{Proof of Proposition~\ref{prop:tv}}
\label{app:tvproof}
Let $\mathcal{O}$ be the space of local observations and let a randomized decision
rule be a measurable $\phi:\mathcal{O}\to[0,1]$, where $\phi(o)$ is the probability
of flagging observation $o$. Fix a fragment position $i$ and write
$P_{\mathrm{att}}=P^{v}_{\mathrm{att},i}$ and $P_{\mathrm{ben}}=P^{v}_{\mathrm{ben},i}$
(a fixed mixture over positions follows by convexity of $\mathrm{TV}$). We have
$\mathrm{TPR}=\int \phi\,dP_{\mathrm{att}}$ and
$\mathrm{FPR}=\int \phi\,dP_{\mathrm{ben}}$. Hence
\begin{align*}
  \mathrm{TPR}-\mathrm{FPR}
  &= \int \phi \,d(P_{\mathrm{att}}-P_{\mathrm{ben}}) \\
  &\le \int \phi \,d(P_{\mathrm{att}}-P_{\mathrm{ben}})^{+} \\
  &\le \int d(P_{\mathrm{att}}-P_{\mathrm{ben}})^{+}
   = \mathrm{TV}(P_{\mathrm{att}},P_{\mathrm{ben}}),
\end{align*}
where the first inequality drops the negative part of the signed measure, the
second uses $0\le\phi\le1$, and the final equality is the definition of total
variation as the total mass of the positive part. The bound holds for every
$\phi$, hence for the best local rule. It is a local distinguishability statement: if benign
and attack local observations are close in TV, no rule on that view separates them beyond that
TV, independent of detector capacity. It does not bound a monitor whose observation
function sees the assembled object, which is exactly the gap assembly monitoring
exploits.

\paragraph{From the bound to AUROC.}
The locality sweep (Section~\ref{sec:localityfailure}) reports per-step AUROC, which relates to
Proposition~\ref{prop:tv} as follows. Writing the ROC curve as $\mathrm{TPR}(\mathrm{FPR})$ traced
by the decision threshold, the area above chance is
\[
  \mathrm{AUROC}-\tfrac12 \;=\; \int_0^1\!\big(\mathrm{TPR}(u)-u\big)\,du ,
\]
where $u=\mathrm{FPR}$. At every threshold Proposition~\ref{prop:tv} bounds the integrand
pointwise by $\mathrm{TV}(P^{v}_{\mathrm{ben}},P^{v}_{\mathrm{att}})$, and $\mathrm{TPR}\le1$, so
\[
  \mathrm{AUROC} \;\le\; \tfrac12+\mathrm{TV}-\tfrac12\mathrm{TV}^2 \;\le\; \tfrac12+\mathrm{TV}.
\]
A local view with small benign--attack TV therefore forces AUROC toward chance, which is what the
sweep measures. AUROC reads in one direction only, and is a diagnostic rather than an estimate
of $\mathrm{TV}$: a high AUROC proves the fragments are distinguishable ($\mathrm{TV}\ge\mathrm{AUROC}-\tfrac12$),
while an AUROC near chance is consistent with small TV but does not prove it.

\subsection{A simple alarm-if-any aggregation bound}
\label{app:coverage}
Proposition~\ref{prop:tv} bounds a single local view. A natural extension is to
aggregate all $K$ per-fragment alarms. We
record a simple aggregation bound showing that this natural alarm-if-any aggregation
does not remove the local-evidence constraint. This is a modest extension of the
single-view bound, not a general result about per-step certification.

Consider a per-step certifier built from a fixed (possibly randomized)
per-fragment rule $\phi:\mathcal{O}\to[0,1]$. At fragment $i$, let
$\alpha_i=\mathbb{E}_{\mathrm{ben},i}[\phi]$ be its benign false-positive rate, and
let $\mathrm{TV}_i=\mathrm{TV}(P^{v}_{\mathrm{ben},i},P^{v}_{\mathrm{att},i})$ be
the local total-variation distance for that fragment. The trajectory-level certifier
raises an alarm if any fragment is flagged.

\begin{proposition}[Alarm-if-any aggregation bound]
\label{prop:coverage}
The trajectory-level false-positive and detection rates satisfy
\[
  \mathrm{FPR}_{\mathrm{traj}}\le \sum_{i=1}^{K}\alpha_i ,
  \qquad
  \mathrm{TPR}_{\mathrm{traj}}\le \sum_{i=1}^{K}\alpha_i+\sum_{i=1}^{K}\mathrm{TV}_i .
\]
Under uniform calibration ($\alpha_i=\alpha$) these become
$\mathrm{FPR}_{\mathrm{traj}}\le K\alpha$ and
$\mathrm{TPR}_{\mathrm{traj}}\le K\alpha+\sum_i\mathrm{TV}_i$. 
Under $\varepsilon$-local indistinguishability
($\mathrm{TV}_i\le\varepsilon$), the detection bound is at most $K(\alpha+\varepsilon)$.
If the benign flag events are additionally independent and
uniformly calibrated, then
$\mathrm{FPR}_{\mathrm{traj}}=1-(1-\alpha)^K$ exactly.
\end{proposition}

\begin{proof}
Let $A_i$ be the event that fragment $i$ is flagged. The trajectory alarm is
$\bigcup_i A_i$. By the union bound,
\begin{multline*}
  \mathrm{FPR}_{\mathrm{traj}}
  = P_{\mathrm{ben}}\!\left(\bigcup_i A_i\right)
  \le \sum_i P_{\mathrm{ben}}(A_i)\\
  = \sum_i \mathbb{E}_{\mathrm{ben},i}[\phi]
  = \sum_i \alpha_i .
\end{multline*}
If the benign flag events $A_i$ are independent and uniformly calibrated, this false
positive rate is exactly $1-(1-\alpha)^K$.

The same argument under the attack distribution gives
\[
  \mathrm{TPR}_{\mathrm{traj}}
  = P_{\mathrm{att}}\!\left(\bigcup_i A_i\right)
  \le \sum_i P_{\mathrm{att}}(A_i)
  = \sum_i \mathbb{E}_{\mathrm{att},i}[\phi].
\]
For each fragment, Proposition~\ref{prop:tv} gives
\[
  \mathbb{E}_{\mathrm{att},i}[\phi]
  \le \mathbb{E}_{\mathrm{ben},i}[\phi]+\mathrm{TV}_i
  = \alpha_i+\mathrm{TV}_i .
\]
Summing over fragments yields
$\mathrm{TPR}_{\mathrm{traj}}\le\sum_i\alpha_i+\sum_i\mathrm{TV}_i$.
\end{proof}

\paragraph{Reading.} The bound separates two sources of trajectory-level alarms. The first is the
per-fragment false-positive budget, $\sum_i\alpha_i$. The second is the actual local
evidence, $\sum_i\mathrm{TV}_i$. Under $\varepsilon$-local indistinguishability, the second
term is small (each $\mathrm{TV}_i\le\varepsilon$): each fragment looks close to benign on its own. The alarm-if-any rule can therefore
increase detection mainly by increasing the per-step alarm budget, which also
increases false alarms.

The last statement makes this tradeoff concrete in the independent-benign-flag
case. If every step has false-positive rate $\alpha$, then the trajectory-level
false-positive rate is $1-(1-\alpha)^K$. For example, setting
$\alpha\approx1/K$ already gives
$\mathrm{FPR}_{\mathrm{traj}}\approx1-e^{-1}\approx0.63$. Thus, auditing more
steps does not by itself solve the problem: without additional local evidence, power
comes from spending more false-positive budget.

This bound applies only to this alarm-if-any aggregation of per-fragment alarms. It
does not rule out monitors that score the fragments jointly. That distinction is the
point: an assembly monitor is not just a union of local alarms. It observes a
different object, and its benign vs.\ attack TV can be large even when every local
marginal $\mathrm{TV}_i$ is small. This natural aggregation still inherits the
local-evidence constraint.

\subsection{Monitor pseudocode}
\label{app:pseudocode}
The monitor taxonomy is in Table~\ref{tab:views} (Section~\ref{sec:setup}); the
pseudocode below summarizes each monitor. All $\tau_\ast$ thresholds are calibrated at
a $5\%$ benign FPR (Appendix~\ref{app:stats}). The per-step, taint, and marker-based
assembly monitors are the ones swept in Figure~\ref{fig:sweep}; the marker-free,
decoded-view, and codec-blind monitors are those analyzed in
Sections~\ref{sec:monitor} and~\ref{sec:boundary}.

\begin{figure*}[t]
\centering
\begin{minipage}{0.62\textwidth}
{\scriptsize\begin{verbatim}
TAINT / PROVENANCE  (whole trace T)
  for each value v from source s to sink k in T:
      if exists path s -> k:  return FLAG
  return PASS                # no visible edge => PASS

PER-STEP STATISTICAL  (one fragment)
  features <- [entropy(x_i), length(x_i), lexical_cues(x_i)]
  return FLAG if f_stat(features) > tau_step

MARKER-BASED DIAGNOSTIC ASSEMBLY  (fragment set {x_1, ..., x_K})
  if any fragment carries the known coordination marker M:
      return 1.0 if fragments assemble   # diagnostic only
  return 0.5                 # chance when marker absent

MARKER-FREE ASSEMBLY  (fragment set, no marker knowledge)
  A <- assembled_text(concat/orderings of {x_i})
  return FLAG if code_structure(A) > tau_asm

DECODED-VIEW GATE  (assembler input, pre-execute)
  for each provided codec family c:
      A <- decode_c(concat fragments)
      if parses_as_executable(A) and matches_exfil_structure(A):
          return BLOCK
  return ALLOW

CODEC-BLIND MONITOR  (fragment set, no decoding step)
  A <- concat(x_i)
  return FLAG if code_structure_density(A) > tau
  # blind: a fully encoded payload scores ~0 and never fires
\end{verbatim}}
\end{minipage}
\caption{Monitor pseudocode. All $\tau_\ast$ thresholds are calibrated at a $5\%$ benign FPR.}
\label{fig:pseudocode}
\end{figure*}

\section{Experimental details and statistics}
\label{app:stats}
\paragraph{Compute and software.}
CPU-only detector experiments use Python with NumPy and scikit-learn. End-to-end ASR runs
serve open-weight instruction models through an OpenAI-compatible endpoint, including vLLM,
on NVIDIA RTX 6000, NVIDIA A100 40GB, and NVIDIA A100 80GB GPUs. Models are served in fp16
where supported; the row named Qwen2.5-32B-AWQ uses the AWQ checkpoint named in the paper.

\paragraph{ASR confidence intervals.}
ASR cells are binomial over $n{=}50$ trajectories. For $0/50$ cells we report the exact
$95\%$ Clopper-Pearson upper bound (0.071); this is a per-cell bound, not a simultaneous
guarantee across cells. For the no-defense rates we report Wilson $95\%$ intervals
(Table~\ref{tab:asr_ci}).

\begin{table}[h]
\centering\small
\caption{Representative ASR confidence intervals ($n{=}50$). The $0/50$ row is an
exact Clopper-Pearson upper bound; the others are Wilson score intervals.}
\label{tab:asr_ci}
\begin{tabular}{lcc}
\toprule
\textbf{ASR} & \textbf{count} & \textbf{$95\%$ CI} \\
\midrule
0.00 & 0/50  & $[0.00,\,0.071]$ \\
0.58 & 29/50 & $[0.44,\,0.71]$ \\
0.76 & 38/50 & $[0.62,\,0.86]$ \\
0.78 & 39/50 & $[0.64,\,0.88]$ \\
1.00 & 50/50 & $[0.91,\,1.00]$ \\
\bottomrule
\end{tabular}
\end{table}

\paragraph{Threshold calibration.}
Runtime gates are calibrated to a $5\%$ benign false-positive rate on held-out
benign trajectories. On Collaborative Shadows this fixes the codec-blind monitor's
threshold at $0.061$; the fully encoded locally benign attack scores $0.013$, below
it, so the codec-blind monitor never fires. The decoded-view gate has a wider
margin: benign trajectories score at most $0.120$ and attacks score $1.240$, so any
threshold in $(0.120,\,1.240)$ gives $0\%$ benign FPR at $100\%$ detection on this
bed (Table~\ref{tab:p8}). This free band is a property of the no-code benign
traffic in the synthetic bed and does not survive code-bearing benign traffic
(Table~\ref{tab:utilitypr}).

\paragraph{Operating-point sensitivity.}
The blindness is not an artifact of the $5\%$ FPR choice. Sweeping the benign-FPR
budget over $\{1\%, 5\%, 10\%\}$ leaves the per-step and per-fragment monitors flat
near zero detection on locally benign attacks: per-step TPR is $0.02$ to $0.03$ across
the range on Collaborative Shadows and $0.00$ for the AgentDojo per-step detector at
every budget, so loosening the threshold does not recover them. The supervised
assembly monitor saturates at TPR $1.000$ at all three budgets. The marker-free
one-class monitor is the exception: its TPR rises with the budget
($0.06 \to 0.25 \to 0.36$), because its strength is ranking (AUROC $0.874$) rather
than a tight operating point; the tight operating-point results therefore rest on
the supervised assembly monitor and the decoded-view gate (used as a diagnostic ceiling, not a deployable defense), not the one-class monitor.

\paragraph{AUROC seeds.}
AUROC values on Collaborative Shadows (Figure~\ref{fig:sweep},
Table~\ref{tab:p3}) are $5$-seed means. The taint and marker-based assembly
checks are rule-based and deterministic given a trajectory, so they take the fixed
values $0.500$ or $1.000$ with no cross-seed variation; seed variation enters only
through the sampled fragment content, and therefore concentrates in the per-step
statistical monitor, whose 5-seed std runs from about $0.01$ where separation saturates
to $0.09$ in the natural-word regimes. On AgentDojo the shipped
neural detector's per-step AUROC carries a per-regime 5-seed std of at most $0.06$, while
the cross-fragment coordination view is deterministic at $1.000$. The marker-free
one-class monitor's AUROC ranges within $0.858$ to $0.879$ across the four held-out
families (mean $0.874$), and the supervised assembly monitor is at $1.000$ on every
held-out family. The qualitative boundary is stable across these seed-level summaries and
across the operating points in Table~\ref{tab:p8}.

\section{Locality sweep details}

\subsection{R-Judge length-artifact details}
\label{app:rjudge}
The length-only classifier reaches AUROC $0.940$, about $96\%$ of the full-content
oracle. The shortcut transfers: the pooled length-only AUROC is $0.940$, and the
large-domain off-diagonal (leave-one-domain-out) AUROC is $0.895$, so the artifact
is not confined to one domain. A length-matched attack drives the deployed
length-exploiting detector to $0.500$ (Cohen's $d=0.000$) while the assembly
view stays at $1.000$. In the length de-confound, broader-observation AUROC drops
from $0.775$ to $0.671$ under length matching; the content-only contrast is $0.618$
versus per-step $0.341$, a net advantage of about $+0.28$. R-Judge is borderline
after de-confounding, which is why the main claim rests on the controllable bed.

\subsection{Full locality sweep}
\label{app:sweep}
Table~\ref{tab:p3} gives the numeric values plotted in Figure~\ref{fig:sweep}.

\begin{table*}[h]
\centering\small
\caption{Locality sweep on Collaborative Shadows (5-seed mean AUROC). Values are
the numeric backup for Figure~\ref{fig:sweep}. \emph{Per-step} is the per-step
statistical monitor; \emph{Assembly} is the diagnostic assembled (marker-based)
view. Note: level 8 is a no-assembly control; the residual assembly score is not
interpreted as recovery (consistent with a residual length cue). Per-step entries give
5-seed mean$\pm$std; taint and assembly are deterministic across seeds (std $0$),
except the level-8 control.}
\label{tab:p3}
\begin{tabular}{c l C{1.3cm}C{2.5cm}C{1.9cm}}
\toprule
\textbf{Level} & \textbf{Encoding} & \textbf{Taint} & \textbf{Per-step} & \textbf{Assembly} \\
\midrule
0 & raw base64            & 1.000 & $0.869\pm.08$ & 1.000 \\
1 & disguised b64         & 1.000 & $0.859\pm.09$ & 1.000 \\
2 & b64 + words           & 0.500 & $0.922\pm.05$ & 1.000 \\
3 & hex                   & 0.500 & $1.000\pm.00$ & 1.000 \\
4 & hex + words           & 0.500 & $0.993\pm.01$ & 1.000 \\
5 & 50\% natural words    & 0.500 & $0.633\pm.09$ & 1.000 \\
6 & 80\% natural words    & 0.500 & $0.705\pm.08$ & 1.000 \\
7 & sentence embed (loc.\ benign) & 0.500 & $0.602\pm.07$ & 1.000 \\
8 & no-assembly control   & 0.500 & $0.621\pm.08$ & $0.857\pm.08$ \\
\bottomrule
\end{tabular}
\end{table*}

\section{Marker-free assembly recovery}
\label{app:markerfree}

\paragraph{Encoding families.} In every family the $K{=}3$ fragments are contiguous
pieces of one real malicious code payload (\texttt{import os, requests \ldots requests.post('http://\allowbreak evil.example/\allowbreak exfil',\allowbreak \ldots)}) spliced into benign report
carriers and length-matched exactly to the benign observation, so no length or
entropy artifact separates attack from benign at the local view. The families differ
only in the coordination marker that wraps each fragment
(Table~\ref{tab:families}); \texttt{F\_none} has no marker and is held out so that a
marker-matching monitor has nothing to key on.

\begin{table*}[h]
\centering\small
\caption{The four marker-free encoding families. All carry the same $K{=}3$
code-payload fragments in benign carriers, length-matched to benign; only the
per-fragment marker differs.}
\label{tab:families}
\begin{tabular}{lll}
\toprule
\textbf{Family} & \textbf{Per-fragment marker} & \textbf{Role} \\
\midrule
\texttt{F\_frag} & \texttt{"fragment $i$:"}            & in-family / training marker \\
\texttt{F\_seg}  & \texttt{"[data segment $i$ of $K$]"} & unseen marker (transfer test) \\
\texttt{F\_ref}  & \texttt{"ref$i$::"}                 & unseen marker (transfer test) \\
\texttt{F\_none} & none                                & held out: no marker to match \\
\bottomrule
\end{tabular}
\end{table*}

\paragraph{Marker-free features.} The assembly monitor scores a seven-feature
vector of the concatenated report texts, with no marker knowledge: (i) code-structure
density, the fraction of code-like tokens (\texttt{import}, \texttt{requests},
\texttt{http://}, \texttt{os.}, \texttt{environ}, brackets, \texttt{=},
\texttt{.post}/\texttt{.get}, \texttt{json}, and similar) in the concatenation;
(ii) partial parseability, obtained by extracting candidate code lines (imports,
assignments, \texttt{requests.*} calls) with a regular expression and attempting
\texttt{ast.parse} on them, scoring $1$ if the extracted block parses and otherwise
$0.2$ per candidate line capped at $1$; (iii, iv, v) the mean, maximum, and standard
deviation of a per-step nonconformity profile from an isolation forest fit on a
\textsc{tf-idf} representation ($300$ features, 1- and 2-grams, $\min\mathrm{df}{=}2$)
of benign step text; (vi) cross-agent length parallelism, the coefficient of
variation of the $K$ report lengths; and (vii) the concatenation length.

\paragraph{Training and evaluation.} The bed is $100$ Collaborative Shadows groups
($K{=}3$), evaluated over $5$ seeds with a $70/30$ group-level train/test split. The
\emph{one-class} monitor standardizes the seven features and fits an isolation forest
($200$ trees, contamination $0.05$) on benign training trajectories only; it
never sees an attack, so its negatives are benign trajectories and its test-time
positives are the attack families. The per-step nonconformity model behind features
(iii, iv, v) is likewise fit on benign training step text. The \emph{supervised} LOFO
monitor trains on the attack families other than the test family
and is evaluated on the held-out family, including \texttt{F\_none}. Separability is
reported as AUROC, which is threshold-free. The isolation forests use an in-fit
contamination of $0.05$ as a training hyperparameter; this is separate from the
runtime gates elsewhere in the paper, which are calibrated to a $5\%$ benign
false-positive rate on held-out benign trajectories (Appendix~\ref{app:stats}).

\paragraph{Results and ablation.} The marker-based monitor scores $1.000$ on its own
family and exactly $0.500$ off-family (mean off-diagonal $0.500$). The one-class
marker-free monitor reaches $0.874$ mean AUROC across all families including
\texttt{F\_none}; trained supervised LOFO it reaches $1.000$ on every
unseen family; the per-step baseline stays at chance ($0.475$ to $0.504$). The
conservative headline number is the one-class $0.874$, because it is never trained on
attacks. Restricting the density feature to
semantic code tokens, dropping whitespace and bracket counts, gives the
\emph{strict} variant used for the natural-language boundary
(Table~\ref{tab:p7}). A leave-one-feature-out ablation identifies code-structure
density as the load-bearing feature for the one-class monitor: dropping it collapses
mean AUROC to $0.465$ (below chance), whereas dropping partial parseability costs only
$0.022$ and dropping any of the per-step nonconformity, length-parallelism, or length
features changes it by at most $0.023$. The supervised monitor is redundancy-robust
(no single drop falls below $0.978$), because density and parseability each separate
the classes on their own. The signal is thus the code-assembly structure of the joint
output, which sharpens why a natural-language payload is the genuine open boundary.

\paragraph{Scope of the per-step result.} The per-step chance result ($0.475$ to $0.504$, 5-seed
mean $0.483$) is for the deployed statistical monitor ($v_{\mathrm{stat}}$,
\textsc{tf-idf}\,+\,isolation-forest). This code bed deliberately gives each fragment a contiguous
slice of the Python payload, so a code-aware per-fragment probe (code-token density plus partial
\texttt{ast} parseability) trivially reads that slice and hits AUROC $1.000$; we report
$v_{\mathrm{stat}}$ precisely to isolate encoding rather than this code-bed artifact. The evidence
that the failure survives a \emph{strong} learned per-fragment detector is the external AgentDojo
check (Section~\ref{sec:localityfailure}): under a real shipped neural detector, locally benign
fragments score at the benign baseline ($0.046$ vs.\ $0.049$).

\begin{table*}[h]
\centering\small
\caption{Assembled-code signal is not a planted marker (numeric backup for
Figure~\ref{fig:lofo}). Marker grep fails off-family; marker-free assembly recovers
held-out families. The deployed statistical per-step monitor ($v_{\mathrm{stat}}$) is at chance; a
code-aware probe ($v_{\mathrm{code}}$) trivially reads the contiguous payload slice on this code
bed (see the paragraph above).}
\label{tab:markerfree}
\begin{tabular}{@{}lcc@{}}
\toprule
\textbf{Monitor} & \textbf{Attacks in training} & \textbf{Held-out AUROC} \\
\midrule
Marker grep & none (rule) & $0.500$ (off-family) \\
Per-step local, statistical $v_{\mathrm{stat}}$ & none & $0.475$ to $0.504$ \\
Per-step local, code-aware $v_{\mathrm{code}}$ & none (rule) & $1.000$ (all families) \\
Marker-free, one-class & none & $0.874$ ($0.858$ to $0.879$) \\
Marker-free, supervised & other families & $1.000$ ($\pm 0$) \\
\bottomrule
\end{tabular}
\end{table*}

\section{End-to-end ASR and representation boundary}

\subsection{Full per-model ASR}
\label{app:asr_full}
Table~\ref{tab:asr_full} gives the full per-model end-to-end numbers behind the
mean-and-band summary in Figure~\ref{fig:asr}. The blindness pattern (which defense
leaves ASR at the no-defense level) is identical across all four models; only the
no-defense height varies with model capability, the contrast shown in
Table~\ref{tab:asr_e2e}.

\begin{table*}[h]
\centering\small
\caption{Entries are ASR under no-defense\,/\,$+$taint\,/\,$+$per-step\,/\,$+$decoded-view
($n{=}50$ per cell, matching the x-axis order of Figure~\ref{fig:asr}). \textbf{Bold}
marks a defense that is blind for that attack family (ASR unchanged from no-defense); the
decoded-view value is $0$ in every cell. The codec-blind monitor is analyzed separately in
Table~\ref{tab:codecblind}; across served models the local-view ASR tracks no-defense while only
the decoded-view gate changes it (Table~\ref{tab:asr_e2e}).}
\label{tab:asr_full}
\begin{tabular}{l C{3.0cm}C{3.0cm}C{3.0cm}}
\toprule
\textbf{Model} & \textbf{ciphertext} & \textbf{clean-word} & \textbf{locally benign} \\
\midrule
Llama-3.1-8B    & 0.58\,/\,0\,/\,0\,/\,0 & 0.58\,/\,\textbf{0.58}\,/\,0\,/\,0 & 0.58\,/\,\textbf{0.58}\,/\,\textbf{0.58}\,/\,0 \\
Qwen2.5-7B      & 0.76\,/\,0\,/\,0\,/\,0 & 0.76\,/\,\textbf{0.76}\,/\,0\,/\,0 & 0.78\,/\,\textbf{0.78}\,/\,\textbf{0.78}\,/\,0 \\
Qwen2.5-14B     & 1.00\,/\,0\,/\,0\,/\,0 & 1.00\,/\,\textbf{1.00}\,/\,0\,/\,0 & 1.00\,/\,\textbf{1.00}\,/\,\textbf{1.00}\,/\,0 \\
Qwen2.5-32B-AWQ & 1.00\,/\,0\,/\,0\,/\,0 & 0.98\,/\,\textbf{0.98}\,/\,0\,/\,0 & 1.00\,/\,\textbf{1.00}\,/\,\textbf{1.00}\,/\,0 \\
\bottomrule
\end{tabular}
\end{table*}

\paragraph{Multi-model codec-blind boundary.}
To check the codec-blind end-to-end boundary is not a single-model artifact, we re-run it on
four model scales. Capability moves the no-defense height, but the boundary is unchanged: the
codec-blind window blocks $0$ attacks on every model, and only the decoded-view gate
blocks (Table~\ref{tab:mm_codecblind}).

\begin{table*}[h]
\centering\small
\caption{Codec-blind end-to-end ASR boundary across model scales (locally benign, $n{=}50$,
shared $5\%$-FPR). Each cell is ASR under no-defense\,/\,$+$per-step\,/\,$+$taint\,/\,$+$codec-blind\,/\,$+$decoded-view.
The codec-blind window leaves ASR at the no-defense level on every model; only the decoded-view
gate drives it to $0$. These rows are an independently regenerated rollout set, so the
no-defense heights differ from Table~\ref{tab:asr_full} because of different generation conditions,
not sampling noise; only the defense pattern is claimed invariant (every local and
codec-blind view equals no-defense, only decoded-view drives it to $0$).}
\label{tab:mm_codecblind}
\renewcommand{\arraystretch}{1.15}
\begin{tabular}{@{}lC{2.0cm}C{2.0cm}C{2.0cm}C{2.0cm}C{2.5cm}@{}}
\toprule
\textbf{Model} & \textbf{none} & \textbf{$+$per-step} & \textbf{$+$taint} & \textbf{$+$codec-blind} & \textbf{$+$decoded-view} \\
\midrule
Qwen2.5-7B   & 0.68 & 0.68 & 0.68 & 0.68 & \textbf{0.00} \\
Llama-3.1-8B & 1.00 & 1.00 & 1.00 & 1.00 & \textbf{0.00} \\
Qwen2.5-14B  & 1.00 & 1.00 & 1.00 & 1.00 & \textbf{0.00} \\
Qwen2.5-32B-AWQ  & 0.98 & 0.98 & 0.98 & 0.98 & \textbf{0.00} \\
\bottomrule
\end{tabular}
\end{table*}

\subsection{ASR provenance and integrity}
\label{app:provenance}
Because the end-to-end ASR depends on the execution runs, we guard model identity
explicitly. Each rollout is saved only after a per-save check that the recorded
\texttt{model} field matches the intended served model, and the runner refuses to
save on a mismatch. All cells in Table~\ref{tab:asr_full} pass this check. Attack-success
rates differ across models because compliance is capability-bound; the monitor-blindness
pattern, by contrast, is invariant across these served models under the fixed codec and
observation views. Mistral-7B was
excluded as a capability floor (near-zero no-defense ASR, so it cannot discriminate
defenses).

\subsection{No-codec LLM search}
\label{app:p16}
This is the no-codec LLM search summarized in Section~\ref{sec:boundary} and added as a row
to Table~\ref{tab:codecblind}. It sharpens the breadth-versus-representation question: can a
capable LLM, rather than the fixed common-decoder ensemble, \emph{discover} the representation
and recover the assembled payload from the fragments without being handed the codec?

\paragraph{Setup.} We reuse the same persisted locally benign rollouts as the codec-blind
boundary cell (Collaborative Shadows, Qwen2.5-14B, $K{=}3$): $50$ attack trajectories and
$50$ benign trajectories. For each trajectory the monitor sees the $K$ fragment texts (with
fragment indices) and nothing about the encoding. We prompt DeepSeek-V4-Flash at temperature
$0$, one call per trajectory, to (i) decide whether the fragments carry a hidden
compositional payload and (ii) propose up to five candidate reconstructions, each a free-form
decoded/assembled string. Every candidate is scored by the \emph{same} gate as the common-decoder
ensemble (code-structure score $\max(\textsc{assembles-to-code},\,10\times\textsc{code-density})$), and
the trajectory score is the maximum over its candidates. The decision threshold is the
$95$th percentile of the benign trajectory scores (the same $5\%$ benign-FPR operating point
used elsewhere); benign trajectories pass through the identical prompt and gate.

\begin{table}[h]
\centering\small
\caption{No-codec LLM search (DeepSeek-V4-Flash, no codec given), scored by
the same gate as the common-decoder ensemble at the shared $5\%$ benign-FPR point. Columns give
the per-class score distribution. At the $5\%$-FPR threshold ($0.191$), the search is flagged
on only $\mathbf{6/50}$ attacks, with AUROC $0.521$ (benign vs.\ attack), essentially
chance.}
\label{tab:p16}
\begin{tabular}{@{}lcc@{}}
\toprule
 & \textbf{attack ($n{=}50$)} & \textbf{benign ($n{=}50$)} \\
\midrule
score mean              & 0.053 & 0.046 \\
score max               & 1.000 & 1.538 \\
nonzero scores          & 7/50  & 5/50 \\
flag rate (detect step) & 0.76  & 0.34 \\
\bottomrule
\end{tabular}
\end{table}

\paragraph{Result.}
The LLM often notices that the trace is suspicious (flag rate $0.76$), but its
candidate reconstructions are flagged on only $6/50$ attacks. Its ranking is
near chance (AUROC $0.521$): the model also produces code-like guesses on benign
traffic, including one benign trajectory with score $1.538$, above the maximum
attack score ($1.000$). A no-codec LLM search therefore narrows the gap only
marginally, remaining much closer to the common-decoder ensemble ($0/50$, AUROC $0.564$) than to
the decoded-view ceiling ($50/50$): representation discovery, not breadth, is the
bottleneck.

\section{Additional boundary and stress tests}
\label{app:opband}
Table~\ref{tab:utilitypr} gives the benign-code utility stress test; Table~\ref{tab:p8} the
full operating-point sweep behind the free operating point; Table~\ref{tab:ksweep} the
$K$-sweep; and Table~\ref{tab:p7} the natural-language payload boundary.
Executable-preserving perturbations, including wrappers, import aliasing, dead code, and
string construction, did not change the boundary: decoded-view access blocked all cases,
while marker-free and common-decoder views missed them. The adaptive
truncation sanity check returned evades-and-executes rate $0.000$ at every retained-code
fraction $f$, which is partly definitional, as explained in Section~\ref{sec:boundary}.

\begin{table}[h]
\centering\small
\caption{Decoded-view gate operating points on the synthetic bed
(Collaborative Shadows locally benign, 50 groups). Benign scores stay at most $0.120$ and
attack scores reach $1.240$, so the bed has a free operating point. Benign FPR is the
gate's utility cost.}
\label{tab:p8}
\setlength{\tabcolsep}{3.5pt}
\begin{tabular}{l ccc}
\toprule
\textbf{operating point} & \textbf{benign FPR} & \textbf{TPR} & \textbf{residual ASR} \\
\midrule
threshold $>0.120$         & 0.000 & 1.000 & 0.000 \\
nominal $5\%$-FPR$^\dagger$       & 0.06 & 1.000 & 0.000 \\
any thr $\in(0.12,1.24)$   & \textbf{0.000} & \textbf{1.000} & \textbf{0.000} \\
\bottomrule
\end{tabular}\\[2pt]
{\footnotesize $^\dagger$The free band (any threshold in $(0.12,1.24)$) already blocks all
attacks at $0$ benign FPR; forcing a nominal $5\%$-FPR threshold instead admits $3/50$ benign
($0.06$) for no detection gain, an artifact of calibrating on the empirically separated scores,
not a real cost.}
\end{table}

\begin{table*}[h]
\centering\small
\caption{\textbf{Code signal is not utility.} On benign coding work, codec-blind
code-structure cues invert the ranking, while the decoded-view gate restores separation by
decoding the assembled representation.}
\label{tab:utilitypr}
\begin{tabular}{l ccccc}
\toprule
\textbf{Monitor} & \textbf{benign mean} & \textbf{attack mean} & \textbf{AUROC} & \textbf{AP} & \textbf{FPR @TPR$.95$} \\
\midrule
Codec-blind code structure & 0.897 & 0.600 & 0.090 & 0.510 & 0.960 \\
Decoded-view gate & 0.897 & 1.240 & 1.000 & 1.000 & 0.000 \\
\bottomrule
\end{tabular}
\end{table*}

\begin{table*}[h]
\centering\small
\caption{Detection-side $K$ sweep at the locally benign operating point. Entries are
AUROC\,/\,detection rate.}
\label{tab:ksweep}
\setlength{\tabcolsep}{5pt}
\begin{tabular}{r cccc}
\toprule
\textbf{$K$} & \textbf{per-step} & \textbf{marker-free raw} & \textbf{common-decoder} & \textbf{decoded-view} \\
\midrule
2  & 0.000 / 0.00 & 1.000 / 0.00 & 0.160 / 0.00 & 1.000 / 1.00 \\
3  & 0.000 / 0.00 & 1.000 / 0.00 & 0.288 / 0.00 & 1.000 / 1.00 \\
5  & 0.000 / 0.00 & 1.000 / 0.00 & 0.456 / 0.00 & 1.000 / 1.00 \\
10 & 0.972 / 0.90 & 1.000 / 0.00 & 0.728 / 0.00 & 1.000 / 1.00 \\
\bottomrule
\end{tabular}\\[2pt]
{\footnotesize Note: high raw AUROC for marker-free code structure does not imply
detection at the calibrated threshold; detections remain $0$ until decoded-view access
(or $K{=}10$ per-step recovery).}
\end{table*}

\begin{table*}[h]
\centering\small
\caption{\textbf{Natural-language payload boundary.} Code-structure monitors do not
transfer cleanly to a non-code harmful instruction. Reported one-class signals are
diagnostic for this single code/NL pair and include formatting artifacts; removing
formatting tokens drops the natural-language score. The code one-class (raw) value here
($0.878$) is for this single code/NL pair; the $0.874$ headline elsewhere is the
leave-one-family-out one-class mean across the four encoding families.}
\label{tab:p7}
\setlength{\tabcolsep}{4pt}
\begin{tabular}{l cccc}
\toprule
\textbf{Payload} & \textbf{assembles?} & \textbf{one-class (raw)} & \textbf{one-class (strict)} & \textbf{marker-based} \\
\midrule
code & 0.56 to 0.66 & 0.878 & 0.956 & 1.000 \\
natural lang. & 0.000 & 0.753$^\ast$ & \textbf{0.406} & 1.000 \\
\bottomrule
\end{tabular}\\[2pt]
{\footnotesize $^\ast$artifact: drops to 0.406 once formatting tokens are removed.}
\end{table*}

\paragraph{Adversarial assembly order and decoys.}
A distributed attacker need not present the fragments in order or in isolation. We
permute the $K$ fragments before assembly and interleave benign decoy fragments
($D\in\{1,3,5\}$). Neither breaks the attack: a per-fragment index lets the decoded-view
gate re-sort the pieces and ignore the decoys, so reconstruction and execution stay at
$1.000$. Detection does not degrade either: the marker-free one-class monitor is
invariant under reordering ($0.873$) and rises under decoys ($0.97$ to $0.99$), while the
supervised monitor and the decoded-view gate remain at $1.000$. Adversarial ordering and
decoy injection therefore do not evade assembly recovery within the scoped threat
model.

\end{document}